\documentclass{amsart}
\usepackage[utf8]{inputenc}
\usepackage{tikz}
\usepackage{graphicx,hyperref,colortbl,amssymb,adjustbox}
\usepackage{tabularx,booktabs,array,mathdots,nicematrix,tikz}
\usepackage{academicons}
\usepackage{xcolor}

\newtheorem{theorem}{Theorem}

\newtheorem{lemma}[theorem]{Lemma}
\newtheorem{definition}[theorem]{Definition}

\title{Multiplayer boycotts in convex games}
\author{Robbert Fokkink \and Hans de Munnik}
\address{Institute of Applied Mathematics\\
Delft University of Technology\\
Mekelweg 4 \\
2628 CD Delft, Netherlands}
\email{r.j.fokkink@tudelft.nl,\ 
j.c.demunnik@student.tudelft.nl
}

\begin{document}

\maketitle

\begin{abstract}
    We extend the notion of boycotts between players in cooperative games to boycotts between coalitions. 
    We prove that convex games offer a proper setting for studying these games. Boycotts have a heterogeneous effect.
    Individual
    players that are targeted by many-on-one boycotts suffer most, while non-participating
    players may actually benefit from a boycott. 
\end{abstract}

During the nineteen-eighties Western nations imposed economic
sanctions against South-Africa. Their effectiveness has been 
debated
, but the apartheid regime did unravel sooner than anticipated~\cite{manby1992south}.
It remains a prime example of the power of boycotts, which are
increasingly being used by Western governments~\cite{meyer2023international}.
Today, economic sanctions are imposed against Iran, North Korea,
the
Russian Federation, and Venezuela. 
We extend previous work on boycotts in cooperative games
and demonstrate the
usefulness of the Shapley value to quantify their impact.

\section{Boycott games}

A cooperative game with transferable utility is a pair
$(N,v)$ where $N$ is the set of \emph{players} and
the \emph{characteristic function} $v$ is 
defined for all \emph{coalitions}
$S\subset N$.
What is the impact if 
one coalition decides to
boycott another coalition? This question has already been
asked and answered by Besner~\cite{besner2022impacts},
for single players.
We extend his work to coalitions. 

\begin{definition}[Besner]
Two players $i,j$ are \emph{disjointly productive} if for all
$S\subset N\setminus\{i,j\}$ we have
\[
v(S\cup\{i,j\})-v(S\cup\{j\})=v(S\cup\{i\})-v(S).
\]
Disjoint coalitions $A, B$ are \emph{disjointly
productive} if all $i\in A$ and $j\in B$
are disjointly productive. 
\end{definition}

We write marginal contribution as
\[
dv_i(S)=v(S\cup\{i\})-v(S)
\]
or more generally
\[
dv_C(S)=v(S\cup C)-v(S).
\]

\begin{lemma}\label{lem:1}
   If $A$ and $B$ are disjointly productive, then for all $A'\subset A$
   and $B'\subset B$ and $S\subset N\setminus(A\cup B)$ 
   \[
   dv_{B'}(S\cup A')=dv_{B'}(S).
   \]
\end{lemma}
\begin{proof} 
    The marginal $dv_{B'}(S\cup A')$ is a sum of marginals of players in $B'$
    that are disjointly productive from $A'$. Therefore, removing
    $A'$ from the coalition does not change the marginal contribution.
\end{proof}

\begin{definition}
    For any $(N,v)$ and any disjoint pair of coalitions $A,B$
    we say that $(N,v^{AB})$ is the \emph{$A,B$-boycott game} if
    \begin{enumerate}
        \item $v^{AB}(S)=v(S)$
    if $S\cap A=\emptyset$ or $S\cap B=\emptyset$,
        \item $A$ and $B$ are
    disjointly productive in $(N,v^{AB})$.
    \end{enumerate}
    If $A=\{i\}$ and $B=\{j\}$ then we write $v^{ij}$
    and we have a \emph{one-on-one boycott}.
    If $A=\{i\}$ and $|B|>1$ then we write $v^{iB}$ and we
    have a \emph{many-on-one boycott}.
\end{definition}
By lemma~\ref{lem:1} we have $v^{AB}(S\cup A'\cup B')-v^{AB}(S\cup B')=v^{AB}(S\cup A')-v^{AB}(S)$.
for $S$ disjoint from $A\cup B$ and $A'\subset A, B'\subset B$.
A rearrangement of terms gives
\begin{equation}\label{convexv}
v^{AB}(S\cup A'\cup B')=v(S\cup A')+v(S\cup B')-v(S),
\end{equation}
which defines~$v^{AB}$.

\begin{definition}
A characteristic function is \emph{supermodular} if \[v(S\cup T)+v(S\cap T)\geq v(S)+v(T)\]
for all coalitions $S,T$.
A cooperative game with a supermodular characteristic function is a \emph{convex game}.    
\end{definition}


\begin{theorem}\label{convexgame}
    $v^{AB}\leq v$ for all disjoint $A$ and $B$
    if and only if $(N,v)$ is convex.
\end{theorem}

\begin{proof}
     This follows immediately from inspection of Equation~\ref{convexv}.
\end{proof}

The purpose of a boycott is a reduction of the opponent's utility,
which is exactly what happens in a convex game.
Also, larger boycotts inflict more harm:



\begin{theorem}\label{lem:3} 
    If $A\subset C$ and $B\subset D$ then $v^{CD}\leq v^{AB}$ for a convex game.
\end{theorem}

\begin{proof}
    For coalitions $S$ and $V$ we write $S_V=S\cap V.$
    Partition a coalition $S$ into $S_A\cup S_{C\setminus A}\cup S_0\cup S_B\cup S_{D\setminus B}$,
    where $S_0=S\setminus(C\cup D)$. We need to prove that $v^{AB}(S)\geq v^{CD}(S)$, which 
    expands into
    {\small
    \begin{equation*}
        \begin{array}{c}
        v(S_A\cup S_{C\setminus A}\cup S_0\cup S_{D\setminus B})+v(S_{C\setminus A}\cup S_0\cup S_{D\setminus B}\cup S_B)-v(S_{C\setminus A}\cup S_0\cup S_{D\setminus B})\\
        \geq\\
        v(S_A\cup S_{C\setminus A}\cup S_0)+v(S_0\cup S_{D\setminus B}\cup S_B)-v(S_0)
        \end{array}
    \end{equation*}
    }
    Converting to marginals and rearranging the terms gives
        \begin{equation*}
        \begin{array}{c}
        dv_{S_A\cup S_{C\setminus A}\cup S_0}(S_{D\setminus B}) +  
        dv_{S_0\cup S_{D\setminus B}\cup S_B}(S_{C\setminus A})\\
        \geq\\
        dv_{S_{C\setminus A}\cup S_0}(S_{D\setminus B})+dv_{S_0}(S_{C\setminus A})
        \end{array}
        \end{equation*}
    As $v$ is supermodular, marginals $dv_X(Y)$
    increase with $X$. Therefore, the left-hand side is indeed larger than the right-hand side.
\end{proof}

Convex games provide the right setting for studying boycotts. 
We now 
determine their impact on individual players.

\section{The impact of a boycott}

A \emph{TU-value} $\varphi$ assigns
a vector $\varphi(N,v)$ to a game $(N,v)$. 
Its coordinates are $\varphi_i(N,v)$, or simply $\varphi_i$, for players $i\in N$.
We write $\varphi(S)=\sum_{i\in S}\varphi_i$. 
It is \emph{efficient} if $\varphi(N)=v(N)$. 
To measure a boycott's impact we compare its value to the value
of the original game. 

\begin{definition}
    The \emph{impact} of a boycott is $\varphi(v)-\varphi(v^{AB}).$
\end{definition}

In a boycott game, $i\in A$ essentially only contributes to $N\setminus B$.
Therefore, it is natural to require that
the value of $i$ in $(N,v^{AB})$ equals its value in $(N\setminus B,v)$. 
We say that such a value is \emph{boycott respecting.}
A value has \emph{balanced impact} if 
in one-on-one boycotts
\[\varphi_i(v)-\varphi_i(v^{ij})=\varphi_j(v)-\varphi_j(v^{ij}).\]
Besner~\cite{besner2022impacts} proved that the Shapley value is
the unique efficient TU-value that
respects boycotts and has balanced impact.
We therefore restrict
out attention to the Shapley value, denoted by $\phi(v)$.

\begin{definition}
For a coalition $C\subset N$, the 
\emph{subgame} $(N,v_C)$ is defined by
\[
v_C(S)=v(S\cap C).
\]
If $\bar A$ denotes the complement of $A$, then the boycott game is
\[
v^{AB}=v_{\bar A}+v_{\bar B}-v_{\bar A\cap \bar B}.
\]
\end{definition}
Since the Shapley value $\phi$ is additive,
the impact of
a boycott is 
\[
\phi(v)-\phi(v_{\bar A})-\phi(v_{\bar B})+\phi(v_{\bar A\cap \bar B}).
\]
A player $i\in A$ 
is a null-player in the subgames on $\bar A$ and $\bar A\cap\bar B$.
For such a player the impact is 
\[
\phi_i(v)-\phi_i(v_{\bar B}).
\]

\begin{theorem}
In a many-on-one boycott $v^{iB}$ the impact
is maximal for $i$.
\end{theorem}
\begin{proof}
     The impact on $i$ is $\phi_i(v)-\phi_i(v_{\bar B})\geq 
     \phi_i(v)-\phi_i(v_{N\setminus \{j\}})$ for any $j\in B$
     since the game is convex.
     The impact on $j\in B$ is $\phi_j(v)-\phi_j(v_{N\setminus \{i\}})$
     which is equal to $\phi_i(v)-\phi_i(v_{N\setminus \{j\}})$
     by balancedness of $\phi$.
     The impact on $i$ is greater than or equal to the impact on $j\in B$.
     
     The impact on a non-participating player $k\in\bar A\cap \bar B$ is 
     \[
     \phi_k(v)-\phi_k(v_{N\setminus \{i\}})-\phi_k(v_{\bar B})+\phi_k(v_{\bar B\setminus \{i\}}).
     \]
     By balancedness this equals
     \[
     \phi_i(v)-\phi_i(v_{N\setminus \{k\}})-\phi_i(v_{\bar B})+\phi_i(v_{\bar B\setminus \{k\}}),
     \]
     which by monotonicity and $\bar B\setminus \{k\}\subset N\setminus \{k\}$ is bounded by the impact on $i$
     \[
     \phi_i(v)-\phi_i(v_{\bar B}).
     \]
\end{proof}

\noindent \textbf{Example.} Non-participating players may benefit. Consider the three-player cooperative game
$v(\{1\})=v(\{2\})=v(\{3\})=0$ and 
$v(\{1,2\})=v(\{1,3\})=v(\{2,3\})=6$ and $v(\{1,2,3\})=12$.
Note that $v(S)=6(|S|-1)$ where $|S|$ denotes
the number of elements. This is a
Myerson communication game on a triangle
with Shapley value $(\phi_1,\phi_2,\phi_3)=(4,4,4)$.
If $1$ boycotts $2$ then $v^{12}(\{1,2\})=0$ but all other values
remain the same. 
The Shapley value of $v^{12}$ is $(3,3,6)$.
Player 3 profits from the boycott. The boycott
deletes the edge between 1 and 2 from the triangle so that 3 becomes central.
\medbreak

Players that participate in a boycott stand to lose.
Players that are unaffected stand to gain. The following
makes that precise.

\begin{definition}
    A player $i$ is \emph{invariant} under a boycott if $v(S)=v^{AB}(S)$ for
    all coalitions that contain $i$.
\end{definition}

\begin{theorem}
    The impact is non-negative for players that participate in
    the boycott and non-positive for invariant players.
\end{theorem}

\begin{proof}
    If $i\in A$ participates in a boycott, then the impact $\phi_i(v)-\phi_i(v_{\bar B})$
    is non-negative by monotonicity.
    If $k$ is invariant, then \[dv^{AB}_k(S)=v^{AB}(S\cup\{k\})-v^{AB}(S)=v(S\cup\{k\})-v^{AB}(S)
    \geq dv_k(S).\]
    Since its marginal contribution is non-decreasing, so is its expected marginal contribution,
    which is equal to the Shapley value. The impact is non-positive.
\end{proof}

\section{The impact of boycotts on trade blocks}

To gain insight in the impact of boycotts on trade networks
we consider some sample games.
Computing the Shapley value of
a network game requires non-trivial algorithms,
which is why we consider simple networks 
known as \emph{block graphs}~\cite{algaba}.
\medbreak
\textbf{A single homogeneous trade block.}
Consider the characteristic function $v(S)=|S|-1$. 
Its Shapley value is $\phi_i(v)=1-\frac {1}n$ for all $i$,
where $n=|N|$. 
If $A$ of size $a$ boycotts $B$ of size $b$ then 
$\phi_i(v^{AB})=1-\frac {1}{n-b}$
for $i\in A$ and 
$\phi_j(v^{AB})=1-\frac {1}{n-a}$ for $j\in B$.
Non-participating
players are invariant and benefit slightly. 
The boycott has minimal impact,
unless a coalition is substantial.
Trade blocks are sheltered against internal trade wars.
\medbreak
\textbf{A single heterogeneous trade block.}
Now there is a special $x\in N$.
We have $v(S)=|S|-1$ if $x\not\in S$ and
$v(S)=3(|S|-1)$ if $x\in S$. The
Shapley value is $\phi_i(v)=2-\frac{1}{n}$ for $i\not =x$
and $\phi_x(v)=n-\frac 1{n}$. In a many-on-one boycott
of $A$ versus $x$, we have
$\phi_i(v^{Ax})=1-\frac{1}{n-1}$ for $i\in A$
and
$\phi_x(v^{Ax})=n-a-\frac{1}{n-a}$.
The impact on non-participating players is negligible.
This situation is similar to a consumer boycott, in which
it is the question if the number of participating consumers 
(whose value halves) is enough 
to make the producer $x$ change its policy~\cite{delacote}.

\medbreak
\textbf{Boycotts between trade blocks.}
Trade networks tend to fall apart into large
internal markets with few outside connections~\cite{wilhite2001bilateral}.
Let $N=I\cup J\cup K$ with equal sized blocks $|I|=|J|=|K|=n$.
The blocks have interconnecting key players $i\in I, j\in J, k\in K$. 
We have $v(S)=|S|$
if there are less than two key players
within $S$.
Key players connect the trade and double 
the value. For instance,
$v(S)=|S|+|S\cap I|+|S\cap K|$ if $i,k\in S$ but $j\not\in S$.
If all three key players are in, then 
$v(S)=3|S|$. The Shapley value is
$\phi_x(v)=\frac 53$ for ordinary players and $\phi_y(v)=\frac 43n+\frac 53$
for key players. Obviously, key players are valueable. 

If $I$ boycotts $J$ then trade between these blocks disintegrates.
The value $v(N)=9n$ reduces to $v^{ij}(N)=7n$.
Trade
block $K$ is unaffected. 
Trade blocks are sheltered against trade wars between other blocks.
The value of non-key players in $I\cup J$ reduces to $\frac 43$,
and for key players it halves
$\frac 23 n+\frac 43$. 
Key players will therefore be hesitant to join. 
Suppose that key player $i$ drops out of the boycott.
Now the ordinary players in $J$ are unaffected and
retain Shapley value $\frac 53$ while players in $I\setminus\{i\}$
remain at $\frac 43$.
The maximum impact is on $i$ and $j$ with value $n+2$.
This boycott hurts $I$ more than $J$.
\begin{table}[h]
\centering
\caption{Russian exports of mineral fuels 2021–23 -- billions of US dollars.
{\tiny Source: Bruegel Russian Foreign Trade Tracker, 2023}.}\label{tab:1}
\begin{tabular}[t]{lccccc}
\hline
Importer&2021&&2022&&2023\\
\cmidrule{2-6}
&spring
&fall&spring&fall&spring\\
\hline
EU27&49.8
&70.6&89.5&61.7&18.5\\
UK&3.5
&3.6&3.2&0.1&0.0\\
China&21.6
&31.0&37.9&45.6&45.9\\
India&1.6
&2.6&9.6&23.7&27.9\\
Turkey&2.5
&3.0&18.8&23.1&15.1\\
\hline
World&97.5
&134.5&178.4&165.0&116.0
\end{tabular}
\end{table}%

\medbreak

Actual trade networks are much more versatile. 
After the Western nations imposed sanctions on the Russian Federation
in 2022, trade from Europe to Russia diverted through members of the
Eurasian Economic Union 
while India became a major consumer of Russian oil~\cite{schott}. 
In February 2023, the EU and the UK stopped imports of crude oil from Russia. 
In anticipation
Russia redirected its trade to the non-participating countries China, India and Turkey, 
compensating for
the impact in 2023,
see Table~\ref{tab:1}. 
The extension of our work to more realistic trade networks is a 
challenging computational task~\cite{skibski}.



\bibliographystyle{siam}
\bibliography{boycott}
\end{document}